\documentclass{tlparxiv}

\let\proof\relax

\usepackage{amsthm}
\usepackage{amsmath}
\usepackage{comment}
\usepackage[ruled,noend,linesnumbered]{algorithm2e}

\usepackage{graphicx}
\usepackage{multirow}
\usepackage{booktabs}
\usepackage{todonotes}
\usepackage{xspace}
\usepackage{tabularx}
\usepackage{longtable}

\usepackage{macros}
\newtheorem{definition}{Definition}

\begin{document}

\lefttitle{\mvpm: Rule-based Reasoning with Uncertain Observations}

\title{
  MV-Datalog+-: Effective Rule-based Reasoning with Uncertain Observations
}

\begin{authgrp}
\author{\sn{Matthias} \gn{Lanzinger}}
\affiliation{University of Oxford}
\author{\sn{Stefano} \gn{Sferrazza}}
\affiliation{University of Oxford}
\affiliation{Technische Universität Wien}
\author{\sn{Georg} \gn{Gottlob}}
\affiliation{University of Oxford}
\end{authgrp}

\maketitle

\begin{abstract}
  Modern applications combine information from a great variety of sources. Oftentimes, some of these sources, like Machine-Learning systems, are not strictly binary but associated with some degree of (lack of) confidence in the observation.

  We propose \mvdl and \mvpm as extensions of Datalog and \datalogpm, respectively, to the fuzzy semantics of infinite-valued \luka logic \Llog as languages for effectively reasoning in scenarios where such uncertain observations occur.
  We show that the semantics of \mvdl exhibits similar model theoretic properties as Datalog. In particular, we show that (fuzzy) entailment can be decided via minimal fuzzy models. We show that when they exist, such minimal fuzzy models are unique (when they exist) and can be characterised in terms of a linear optimisation problem over the output of a fixed-point procedure. 
  On the basis of this characterisation, we propose similar many-valued semantics for rules with existential quantification in the head, extending \datalogpm. 
  This paper is under consideration for acceptance in TPLP.
\end{abstract}


\section{Introduction}
Datalog and its extensions have long been important languages in databases and at the
foundation of many rule-based reasoning formalisms. Continuous
theoretical and technical improvements have led to the development of
highly efficient Datalog systems  for widespread practical use in a variety of applications (see e.g,~\cite{DBLP:books/mc/18/MaierTKW18}).

However, in many such real-world applications the observations and facts that serve as the input are not actually certain but rather associated with some (possibly unknown) level of uncertainty. One particularly important situation where such a setting occurs are settings where observations are made by Machine Learning (ML) systems.

For example, say a database contains a relation that labels images of animals with the class of animal observed in the image. In modern settings, such labels are commonly derived from the output of ML systems that attempt to classify the image contents. Generally, such systems output possible labels for an image, together with a score that can be interpreted as the level of confidence in the labels correctness (cf.,~\cite{DBLP:journals/cviu/MishkinSM17,DBLP:journals/corr/abs-1709-03806}).

When the labelling in our example is used in further logical inference, the level of confidence in the individual labels would ideally also be reflected in any knowledge inferred from these uncertain observations. However, classical logical reasoning provides no way to consider such information and requires the projection of uncertainty levels to be either true or false, usually via some simple numerical threshold. Beyond the overall loss of information, this can also lead to problematic outcomes where logical conclusions are true despite their derivation being based on observations made with only moderate confidence. With observations from ML systems becoming a commonplace feature of many modern reasoning settings, the projection to binary truth values is severely limiting the immense potential of integrating ML observations with logic programming.

\paragraph{Related Work}
There are two natural ways to interpret uncertain observations as
described in the example above. One can consider them as the
likelihood of the fact being true, i.e., the level of confidence in a
fact is interpreted as a probability of the fact holding
true. Such an interpretation has been widely studied in the context of Problog~\citep{DBLP:conf/ijcai/RaedtKT07}, Markov Logic Networks~\citep{DBLP:journals/ml/RichardsonD06}, and probabilistic databases~\citep{DBLP:series/synthesis/2011Suciu,DBLP:conf/vldb/CavalloP87}.  However, generally such formalisms make strong assumptions on the pairwise probabilistic independence of all tuples which can be difficult to satisfy. An approach to probabilistic reasoning in \datalogpm based on the chase procedure was proposed very recently by~\cite{DBLP:conf/ruleml/BellomariniLSS20}.

Alternatively, one can express levels of confidence in a fact in
terms of \emph{degrees of truth} as in \emph{fuzzy logics}
(cf. \cite{DBLP:books/kl/Hajek98}). That is, a fact that is
considered to be true to a certain degree, in accordance with the
level of confidence in the observation. 

There have been a wide variety of approaches of combining logic programming with fuzzy logic.
Pioneering work in this area by~\cite{DBLP:journals/actaC/AchsK95}, \cite{DBLP:journals/fss/Ebrahim01}, and \cite{DBLP:journals/fss/Vojtas01} introduced a number of early foundational ideas for fuzzy logic programming. 
The first also considers fuzziness in the object domain (expressing the possibility that two constants may actually refer to the same object), via fuzzy similarity relations on the object domain. This idea was further expanded in the context of deductive databases by~\cite{DBLP:journals/tfs/IranzoS18}.

One limiting factor of these approaches in our intended setting is the reliance on the \emph{G\"odel t-norm} (see, \cite{DBLP:conf/lpar/Preining10}) as the basis for many-valued semantics. The G\"odel t-norm of two formulas $\phi$, $\psi$ is the minimum of the respective truth degrees of $\phi$ and $\psi$. While this simplifies the model theory of the resulting logics, the resulting semantics are not ideal for our envisioned applications in which we want to combine uncertain observations from AI systems with large certain knowledge bases. Further discussion of the important differences between semantics based on the G\"odel t-norm and the \luka t-norm (see,~\cite{DBLP:books/kl/Hajek98}), which we propose to use in this paper, is provided in Section~\ref{sec:prelims}.

Finally, Probabilistic Soft Logic (PSL) was introduced by~\cite{DBLP:journals/jmlr/BachBHG17} as a framework for fuzzy reasoning with semantics based on \luka logic. PSL is a significantly different language with different use cases than our proposed languages. It allows for negation in rule bodies as well as disjunction in the head, in addition to certain forms of aggregation. Semantically it treats all rules as soft constraints and aims to find interpretations that satisfy rules as much as possible. Similar ideas have also been proposed in fuzzy description logics, see e.g.,~\cite{DBLP:journals/kbs/BobilloS16,DBLP:conf/sum/BorgwardtP17,DBLP:journals/jair/StoilosSPTH07}, or \cite{DBLP:journals/ws/LukasiewiczS08} for an overview of fuzzy DLs.

\paragraph{Contributions}
We propose extensions of Datalog and \datalogpm~\cite{},i.e., Datalog with existential quantification in the head, to the many-valued semantics of infinite-valued \luka logic. We argue that these languages offer natural semantics for practical applications with uncertain data while preserving key theoretical properties of the classical (two-valued) case. 
Overall, the contributions of this paper are summarised as follows.

\begin{itemize}
\item We introduce \mvdl as Datalog with \luka semantics. We discuss adaptions of the standard minimal model concept to fuzzy models where every fact takes the least possible truth degree. We show that for consistent inputs, unique minimal models always exist and give a characterisation of such minimal models.
Our results hold also for $K$-fuzzy models, i.e., interpretations that satisfy formulas of the program only to degree $K$ rather than fully.
\item
    We show that \mvdl is suitable for practical applications. For one, our characterisation reveals that finding $K$-fuzzy models is feasible by building on standard methods as from the classical setting. In particular, we also show fuzzy fact entailment in \mvdl has the same complexity as fact entailment in Datalog.
  \item
    We propose \mvpm as practice oriented extension of \mvdl by existential quantification in the head. To obtain intuitive semantics for logic programming we deviate from standard semantics of many valued existential quantification and instead argue for the use of strong existential quantification. We discuss fundamental model theoretic properties and challenges of \mvpm and show how our characterisation result for \mvdl can be extended to capture a natural class of preferred models of \mvpm, providing the foundation for a system implementation of \mvpm.
\end{itemize}

\paragraph{Structure}
The rest of this paper is structured as follows. Terminology, notation and further preliminaries are stated in Section~\ref{sec:prelims}. Throughout Section~\ref{sec:mvdl}, we introduce \mvdl and our main results for minimal fuzzy models in \mvdl. The further extension to \mvpm is presented in Section~\ref{sec:mvpm}.
We conclude with final remarks and a brief discussion of future work in Section~\ref{sec:conclusion}.
Proof details that had to be omitted from the main text due to space restrictions are provided in~\ref{sec:proofs}.

\section{Preliminaries}
\label{sec:prelims}

We assume familiarity with standard notions of Datalog and first-order logic. Throughout this document we consider only logical languages with no function symbols. 
\paragraph{Datalog}
Fix a relational signature $\sigma$ and a (countable) domain $\dom$.
A \emph{(relational) atom} is a term $R(u_1,\dots,u_{\#R)})$ where $\#R$ is the arity of relation symbol $R\in \sigma$ and for $0 \leq i \leq \#R$, $u_i$ is a variable or a constant. If all positions of the atom are constants we say that it is a \emph{ground atom} (or \emph{fact}).
A \emph{(Datalog) rule} $r$ is a first-order implication $\varphi \rightarrow \psi$ where $\varphi$ is a conjunction of relational atoms and $\psi$ consists of a single relational atom. We say that $\varphi$ is the \emph{body} of $r$ ($\body(r)$) and $\psi$ is the \emph{head} of $r$ ($\head(r)$). A rule is \emph{safe} if all free variables in the head of a rule also occur in its body.
A \emph{Datalog program} is a set of safe Datalog rules. The models of a Datalog program are the first-order interpretations that satisfy all rules. For two models $\I, \I'$ of a Datalog program, we say that $\I \leq \I'$ if every every ground atom satisfied by $\I$ is also satisfied by $\I'$. The \emph{minimal model} of a Datalog program $\Pi$ is the model $\I$ of $\Pi$ such that for all other models $\I'$ of $\Pi$ it holds that $\I \leq \I'$.

A \emph{database} $D$ is a set of ground atoms, logically interpreted as a theory consisting of an atomic sentence for each ground atom. We say that $\I$ is a model of a program and a database $(\prog, D)$, if $\I$ is a model of both $\prog$ and $D$. We write $\prog,D \models \varphi$ if formula $\varphi$ is satisfied for every model of $(\prog, D)$. Note that this is equivalent to satisfaction under the minimal model of $(\prog, D)$. 

\paragraph{\datalogpm}
The extension of Datalog with existential quantification in the head, led to the well studied \datalogpm~\cite{DBLP:journals/ws/CaliGL12}.
A \emph{\datalogpm rule} is either a Datalog rule or a first-order formula of the form $\varphi \rightarrow \exists \ybf\ \psi(\ybf)$ where $\varphi$ is again a conjunction of relational atoms and $\psi$ is a relational atom that mentions all variables in $\ybf$. A \emph{\datalogpm program} is a set of safe \datalogpm rules. Models in \datalogpm are defined in the same way as for plain Datalog. Note that existential quantification ranges over the whole (infinite) domain $\dom$ rather than the active domain of the program/database. It follows that, in contrast to Datalog, \datalogpm programs do not have a unique minimal model.
Furthermore, fact entailment in \datalogpm is undecidable in general but vast decidable fragments are known. 

\paragraph{Oblivious Chase}
Fact entailment in \datalogpm (and Datalog) is closely linked to the so-called \emph{chase} fixed-point procedures. In the following we will be particularly interested in a specific chase procedure called the \emph{oblivious chase}. We give only the minimal necessary technical details here and refer to \cite{DBLP:journals/jair/CaliGK13} for full details.

For any formula $\varphi$, we write $\vars(\varphi)$ for the set of \emph{free variables} in the formula.
Given a database $D$, we say that a \datalogpm rule $r$ is \emph{obliviously applicable} to $D$ if there exists a homomorphism $h \colon \vars(r) \to \dom$ such that all atoms in $h(\body(r))$ are in $D$. 

Let $r$ be a \datalogpm rule that is obliviously applicable to $D$ via a homomorphism $h$. Let $h'$ be the extension of $h$ to the existentially quantified variables $\ybf$ in $\head(r)$ such that for each $y \in \ybf$, $h'(y)$ is a fresh constant, called a \emph{labelled null}, not occurring in $D$ ($h'=h$ when $r$ is a Datalog rule). Let $\psi$ be the atomic formula in $\head(r)$.  The result of the \emph{oblivious application} of $r, h$ to $D$ is $D' = D \cup h'(\psi)$. We write $D \xrightarrow[]{r,h} D'$. 
Let $D$ be a database and $\prog$ a \datalogpm program. The \emph{oblivious chase sequence} $\mathit{OChase}(\prog, D)$ is the sequence $D_0, D_1, D_2, \dots$ with $D=D_0$ such that for all $i \geq 0$, $D_i \xrightarrow[]{r_i,h_i} D_{i+1}$ and $r_i \in \prog$, the chase sequence proceeds as long as there are oblivious applications $r,h$ that have not yet been applied. We refer to the limit of the oblivious chase sequence as $\mathit{OLim}(\prog, D)$. Every oblivious application $r,h$ induces a grounding of $r$ where all variables (including the existentially quantified ones) are assigned constants according to $h$ and any existential quantification is removed. We refer to the set of all ground rules induced in this way by the applications in $\mathit{OChase}(\prog, D)$ as $\oground(\prog, D)$. Note that for every $\prog$ and $D$, $\mathit{OLim}(\prog, D)$ is unique up to isomorphism.

\paragraph{\luka Logic}
While there are many semantics for many-valued logic, in this paper we will consider infinite-valued \luka logic \Llog (see,~\cite{DBLP:books/kl/Hajek98}). Some of our techniques are particular to \luka logic and the results given in this paper generally do not apply to other many-valued logics.
For some relational signature $\sigma$, we consider the following logical language, where
$R$ is a relational atom (in $\sigma$) and a \emph{formula} $\varphi$ is defined via the grammar
\[
  \varphi ::= R \mid \varphi \conj \varphi \mid \varphi \disj \varphi \mid \varphi \rightarrow \varphi \mid \neg \varphi
\]

For a signature $\sigma$ and domain $\dom$, let $\gatom$ be the the set of all ground atoms with respect to $\sigma$ and $\dom$. A \emph{truth assignment} is a function $\nu \colon \gatom \to [0,1]$, intuitively assigning a degree of truth in the real interval $[0,1]$ to every ground atom. 
By slight abuse of notation we also apply $\nu$ to formulas to express the truth of ground formulas $\gamma$, $\gamma'$ according to the following inductive definitions.
\[
  \begin{array}{ll}
  \nu(\neg \gamma) &= 1 - \nu(\gamma) \\
  \nu(\gamma \conj \gamma') &= \max \{0, \nu(\gamma) + \nu(\gamma') -1 \} \\
  \nu(\gamma \disj \gamma') &= \min \{1, \nu(\gamma) + \nu(\gamma') \} \\
  \nu(\gamma \rightarrow \gamma')& = \min \{1, 1 - \nu(\gamma) + \nu(\gamma')\}
  \end{array}
\]
That is, $\conj$ is the usual \luka \textit{t-norm} and $\disj$ is the corresponding \textit{t-conorm}, which take the place of conjunction and disjunction, respectively. Implication $\gamma \rightarrow \gamma'$ could equivalently be defined as usual as $\neg \gamma \disj \gamma'$. Note that De Morgan's rule holds as usual for $\conj$ and $\disj$.

For rational $K \in (0,1]$ we say that a formula $\varphi$ is \emph{$K$-satisfied} by $\nu$ if for every grounding $\gamma$ of $\varphi$ over $\dom$ it holds that $\nu(\gamma) \geq K$. Whenever we make use of $K$ in this context we  assume it to be rational.
In the context of rule-based reasoning it may be of particular interest to observe that an implication is $1$-satisfied exactly when the head is at least as true as the body. For a set of formulas $\Pi$, we say that a truth assignment $\nu$ is a \emph{$K-$fuzzy model} if all formulas in $\Pi$ are $K$-satisfied by $\nu$.

In place of the database in the classical setting, we instead consider \emph{(finite) partial truth assignments}, that is,  partial functions $\tau: GAtoms \to (0,1]$ that are defined for a finite number of ground atoms. 
Let $(\prog, \tau)$ be a pair where $\Pi$ is a set of formulas and $\tau$ is a partial truth assignment, a \emph{$K$-fuzzy model}  of $(\Pi, \tau)$ is a $K$-fuzzy model $\nu$ of $\prog$ where $\nu(G) = \tau(G)$ for every ground atom $G$ for which $\tau$ is defined. Whenever we talk about formulas $\prog$ and partial truth assignments $\tau$ we use $\mathit{ADom}$ to refer to their \emph{active domain}, i.e., the subset of the domain that is mentioned in either $\prog$ or $\{G \in \gatom \mid \tau(G) \text{ is defined.}\}$. We write $\activeatoms$ to indicate $\gatom$ restricted to groundings over $\mathit{Adom}$.

\emph{$K$-fuzzy models} and their theory have been introduced and studied in the context of logic programming by~\cite{DBLP:journals/fss/Ebrahim01} but under different many-valued semantics based on the G\"odel t-norm.
\begin{example}
    Consider a labelling of images as described in the introduction, represented by predicate $\mathit{Label}$. Additionally, we have a predicate corresponding to whether the image was taken in a polar region. Suppose we have the following observations for image $i_1$.
  \[
    \tau(\mathit{Label}(i_1, \mathrm{Whale})) = 0.8 \qquad \tau(\mathit{PolarRegion}(i_1)) = 0.7
  \]
  Now we consider satisfaction of the following ground rule under the semantics studied by~\cite{DBLP:journals/fss/Ebrahim01} in contrast to \luka semantics.
  \[
    \mathit{Label}(i_1, \mathrm{Whale}) \conj  \mathit{PolarRegion}(i_1) \rightarrow \mathit{Orca}(i_1)
  \]
  If we interpret implication and $\conj$ under the semantics proposed by~\citet{DBLP:journals/fss/Ebrahim01}, the body has truth $0.7$ (the G\"odel t-norm of $0.8$ and $0.7$). The implication evaluates to the maximum of $0.3$ ($=1-$ truth of the body) and the truth of the head. Hence, a model that $1$-satisfies the rule would have to satisfy $\mathit{Orca}(i_1)$ with truth degree $1$.
  This can be unintuitive, especially in the context of logic programming, since the truth degree of the fact is not inferred from the truth of the body.

  In \luka semantics, the body has truth degree $0.5$ following the intuition that the conjunction of two uncertain observations becomes even less certain. For the implication to be $1$-satisfied it suffices if $\mathit{Orca}(i_1)$ is true to at least degree $0.5$. We see that the inferred truth degree of the head naturally depends on the truth of the body.
\end{example}

\section{\mvdl and Minimal Fuzzy Models}
\label{sec:mvdl}
In this section we propose the extension of Datalog to \luka semantics for many-valued reasoning. We show that \mvdl allows for an analogue of minimal model semantics in the fuzzy setting. Moreover, we give a characterisation of such minimal models in terms of an optimisation problem over the ground program corresponding to the oblivious chase.

\subsection{\mvdl}

\label{sec:min}
\begin{definition}[\mvdl Program]
  An \emph{\mvdl program} $\prog$ is a set of \Llog formulas of the form
  \[
    B_1 \conj B_2 \conj \dots \conj B_n \rightarrow H
  \] where all $B_i$, for  $1 \leq i \leq n$, and $H$ are relational atoms.
\end{definition}
As noted in Section~\ref{sec:prelims}, we consider a partial truth assignment $\tau$ in the place of a database. We will therefore also refer to $\tau$ as the database. We call a pair $\prog, \tau$ of an \mvdl program and database an \emph{\mvinst}.
We can map a \mvdl program $\prog$ naturally to a Datalog program by substituting $\conj$ by $\land $. We refer to the resulting Datalog program as  $\progcrisp$. For the respective crisp version of $\tau$ we write $D_\tau$ for the database containing all facts $G$ for which $\tau(G)$ is defined and greater than $0$.

Analogous to fact entailment in classical Datalog, we have the central decision problem \ktruth, whether a fact is true to at least a degree $c$ in all models, as follows.
\begin{problem}{\normalsize{\ktruth}}
  Input & \mvinst $(\prog, \tau)$, ground atom $G$, $c \in[0,1]$\\
  Output& $\nu(G) \geq c$ for all $K$-fuzzy models $\nu$ of $(\prog, \tau)$?
\end{problem}
\mvdl is a proper extension of Datalog in the sense that for $K=1$ and when all ground atoms in the database are assigned truth $1$, its models coincide exactly with those of Datalog programs.

Note that in contrast to Datalog, \mvinsts do not always have models. For example, consider a program consisting only of rule $R(x) \rightarrow S(x)$. For any $K > 0$, there is a database $\tau$ that assigns truth $1$ to $R(a)$ and some truth less than $K$ to $S(a)$ such that the rule is not $K$-satisfied.
This is a consequence of the definition of a $K$-fuzzy model $\nu$ of $(\prog, \tau)$, which requires the truth values in $\nu$ to agree exactly with $\tau$, for every fact for which $\tau$ is defined. In some settings it may be desirable to relax this slightly and consider $K$-fuzzy models $\nu$ where $\nu(G) \geq \tau(G)$ where $\tau(G)$ is defined\footnote{Considering $\nu(G) \leq \tau(G)$ allows the trivial model setting everything to truth $0$}.
Our semantics cover such a relaxation since it can be simulated by straightforward rewriting of the program: for every relation symbol $R$ that occurs in $\tau$ add rule $R(\xbf) \rightarrow R'(\xbf)$, and replace all other occurrences of $R$ in the program by $R'$.

\citet{DBLP:journals/fss/Ebrahim01} introduced minimal  $K$-fuzzy models to study logic programming with the G\"odel t-norm, as the intersection of all $K$-fuzzy models, where truth values of an intersection are taken to the minimum between the two values. We will use the following alternative (but equivalent) definition here that will be more convenient in our setting. For two truth assignment $\nu$, $\mu$ we  write $\nu \leq \mu$ when for all $G \in \gatom$, it holds that $\nu(G) \leq \mu(G)$. We similarly write $\nu < \mu$ if $\nu \leq \mu$ and $\nu(G)<\mu(G)$ for at least one $G \in \gatom$.

\begin{definition}
  Let $\prog, \tau$ be an \mvinst. We say that a $K$-fuzzy model $\mu$ of
  $(\prog, \tau)$ is \emph{minimal} if for every $K$-fuzzy model $\nu$
  of $(\prog, \tau)$ it holds that $\mu \leq \nu$.
\end{definition}

 Minimal $K$-fuzzy models behave very similar to minimal models in Datalog. In particular, the two following propositions state two key properties allows us to focus on a single minimal fuzzy model for the problem \ktruth.

\begin{theorem}
  \label{thm:main}
  Let $\prog$, $\tau$ be an \mvinst. For every rational $K\in (0,1]$, if $(\prog, \tau)$ is $K$-satisfiable, then there exists a unique minimal $K$-fuzzy model for $(\prog, \tau)$.
\end{theorem}
\begin{proof}
  We show that for every set of $K$-fuzzy models $S$, the infimum of $S$, defined as
  the truth assignment $\nu'$ with $\nu'(G) = \inf_{\nu\in S} \nu(G)$ for all $G \in \gatom$, is also a $K$-fuzzy model for $(\prog, \tau)$.

  For any ground rule $\gamma$ and $\nu \in S$, the body truth value $\nu(\body(\gamma))$ is given by $\max\{0, \left(\sum_{i=1}^\ell \nu(G_i)\right) - \ell + 1\}$ where the body of $\gamma$ consists of ground atoms $G_1,\dots,G_\ell$. Hence, since we have $\nu'(G_i)\leq \nu(G_i)$ for every $1\leq i\leq \ell$ by construction, it also holds that $\nu'(\body(\gamma))\leq \inf_{\nu \in S} \nu(\body(\gamma))$.

  Therefore also $\nu'(\body(\gamma)) \leq (1-K)+\inf_{\nu \in S} \nu(\head(\gamma))$, since all $\nu \in S$ are $K$-fuzzy models.
  That is, there exists a rational $r = \max\{0, \nu'(\body(\gamma))-1+K\}$ (recall that we only consider rational $K$) such that $r$ is less or equal to $\nu(\head(\gamma))$ for every $\nu \in S$ and $\nu'(\body(\gamma)) \leq (1-K)+  r$.
    Then, $\nu'(\head(\gamma)) = \inf_{\nu \in S} \nu(\head(\gamma)) \geq r$ and therefore also $\nu'(\body(\gamma)) \leq (1-K)+  \nu'(\head(\gamma))$, i.e., $\nu'$ $K$-satisfies $\gamma$.
    We see that $\nu'$ is indeed a $K$-fuzzy model and clearly for every $\nu \in S$ we also have $\nu' \leq \nu$. Taking the infimum over all $K$-fuzzy models of $(\prog, \tau)$ then yields the unique minimal $K$-fuzzy model for $(\prog, \tau)$.
\end{proof}

\begin{proposition}
  \label{prop:obv}
  Let $G$ be a ground atom, $c \in [0,1]$, and $\mu$ the minimal $K$-fuzzy model of $(\prog, \tau)$. Then $\mu(G) \geq c$ if and only if $\nu(G) \geq c$ for all $K$-fuzzy models of $(\prog, \tau)$.
\end{proposition}

Of particular importance to our main result is the fact that fuzzy models are, in a sense, only more fine-grained versions of classical models. In particular, since the classical case corresponds to every observation having maximal truth degree, fuzzy models are in a sense bounded from above by classical models. 
\begin{lemma}
  \label{lem:crispub}
  Let $(\prog, \tau)$ be an \mvinst. Let $\nu$ be a minimal $K$-fuzzy model of $(\prog, \tau)$. Then for every $G \in \gatom$, we have $\nu(G)>0$ only if $\progcrisp, D_\tau \models G$.
\end{lemma}

\subsection{Characterising Minimal Fuzzy Models}

Minimal $K$-fuzzy models are a natural tool for query answering. In the following we show that we can in fact characterise the minimal $K$-fuzzy model for $(\prog, \tau)$ in terms of a simple linear program $\optk$ defined over the ground rules induced by the oblivious chase sequence\footnote{While the use of the oblivious chase without existential quantification is untypical we explicitly require all rules induced by oblivious applications for our characterisation.} for $\progcrisp, D_\tau$.

Let $(\prog, \tau)$ be an \mvinst.
Let  $\Gamma=\{\gamma_1, \dots, \gamma_m\}$ be the ground rules in $\oground(\progcrisp, D_\tau)$. Let $\mathcal{G}=\{G_1,\dots,G_n\}$ be all of the ground atoms occurring in rules in $\Gamma$.
For every $G_i \in \mathcal{G}$ we associate $G_i$ with a variable $x_i$ in our linear program that intuitively will represent the truth degree of $G_i$.
For $\gamma_j$ of the form $G_{j_1},\dots,G_{j_\ell} \rightarrow G_{j_\mathit{head}}$ define
\[
  \Vf(\gamma_j) := \sum_{k=1}^\ell \left(1-x_{j_k} \right)+  x_{j_{\mathit{head}}} 
\]
which directly expresses the satisfaction of rule $\gamma_j$, with variable $x_{j_k}$ representing the truth of $G_{j_k}$.
The linear program $\optk$ is then defined as follows 
\begin{equation}
  \begin{array}{llr}
    \text{minimise} & \sum_{i=1}^n x_i  \\
    \text{subject to} 
                    & \Vf(\gamma_j) \geq K & \text{for } 1 \leq j \leq m\\
                    & x_i = \tau(G_i) & \text{for $i$ where $\tau(G_i)$ is defined}\\
                    & 1 \geq x_i \geq 0& \text{for } 1 \leq i \leq n
  \end{array}
  \label{eq:baselp}
\end{equation}

Recall, a \emph{solution} of a linear program is any assignment to the
variables, a \emph{feasible solution} is a solution that satisfies all
constraints, and an \emph{optimal solution} is a feasible solution
with minimal value of the objective function. With this in mind we can state our main result for \mvdl, a characterisation of minimal $K$-fuzzy models in terms of optimal solutions of $\optk$.
\begin{theorem}
  \label{thm:key}
  Let $\prog$, $\tau$ be an \mvinst. Then $\xbf$ is an optimal solution of $\optk$ if and only if $\nu_\xbf$ is a minimal $K$-fuzzy model of $(\prog,\tau)$
\end{theorem}
By construction, any feasible solutions of \optk induces a $K$-fuzzy
model $\nux$ that assigns $\nux(G_i)=x_i$ and $0$ to all other ground
atoms not in $\mathcal{G}$. Similarly, in the other direction we can
observe that any $K$-fuzzy model $\nu$ of $(\prog, \tau)$ can be
used to construct a feasible solution of $\optk$. This natural correspondence of solutions of $\optk$ and $K$-fuzzy models of $(\prog, \tau)$ is formalised in the following lemma.
\begin{lemma}
  \label{lem:opt}
  Let $\prog$, $\tau$ be an \mvinst. The following statements hold:
  \begin{enumerate}
  \item For any feasible solution $\xbf$ of $\optk$, $\nux$ is a $K$-fuzzy model of $(\prog, \tau)$.
  \item For any $K$-fuzzy model $\nu$ of $(\prog, \tau)$, the solution $\xbf$ with $x_i = \nu(G_i)$ is a feasible solution of $\optk$.
  \item
       $\optk$ is feasible if and only if $(\prog, \tau)$ is $K$-satisfiable.
  \end{enumerate}
 \end{lemma}
We say that a ground rule
$\gamma \in \Gamma$ is \emph{$\nu$-tight} if $\nu(\gamma)=K$, that is
$\nu(\body(\gamma))-1+K = \nu(\head(\gamma))$. When a rule is not tight we refer to $\nu(\head(\gamma)) - (\nu(\body(\gamma))-1+K)$ as the \emph{$\nux$-gap} of the rule. The gap is at least $0$ for any $K$-satisfied ground rule. We will be interested
in the structural relationships between tight rules.
We refer to the ground atoms in $\mathcal{G}$ for which $\tau$ is not defined as the \emph{derived ground atoms} of $\optk$.
\begin{lemma}
  \label{lem:cycle}
  Let $\nu$ be a minimal model of an \mvinst $\prog, \tau$. Let $\mathcal{G}'$ be a non-empty set of ground atoms for which $\nu$ assigns truth greater than $0$.
  Then there exists at least one $\nu$-tight ground rule in $\oground(\progcrisp, D_\tau)$ such that $\head(\gamma)\in \mathcal{G}'$ and $\body(\gamma)\cap \mathcal{G}' = \emptyset$.
\end{lemma}

\begin{proof}[Proof of Theorem~\ref{thm:key}]
  Observe, that by Lemma~\ref{lem:crispub} and the fact that the chase produces all atoms that can be entailed by $(\prog, \tau)$, it follows that we do not have to consider solutions with atoms other than the set $\mathcal{G}$ considered in the definition of $\optk$.

  For the if direction, let $\mu$ be a minimal model of $(\prog, \tau)$ with $\mu(G) = 0$ for all ground atoms $G \not \in \mathcal{G}$. By Lemma~\ref{lem:opt}, $\mu$ induces a feasible solution $\xbf$ of $\optk$ where $x_i = \mu(G_i)$. Suppose this solution is not optimal and let $\ybf$ be an optimal solution of $\optk$. Let $\mathcal{G}' \subseteq \mathcal{G}$ be the set of ground atoms $G$ where $\nuy(G) < \mu(G)$. Since we assume that $\sum y_i < \sum x_i$, this set is not empty. By Lemma~\ref{lem:cycle} there exists a $\mu$-tight rule $\gamma$ where the head is in $\mathcal{G}'$ but none of the body atoms of $\gamma$ are. It follows that
  \[
    \nuy(\head(\gamma)) < \mu(\head(\gamma)) = \mu(\body(\gamma)) -1+K= \nuy(\body(\gamma)) -1+K
  \]
  Hence, $\nuy$ does not $K$-satisfy $\gamma$. By Lemma~\ref{lem:opt} this contradicts the assumption that $\ybf$ is a feasible solution with lower objective than $\xbf$.

  For the other direction, let $\xbf$ be an optimal solution of
  $\optk$. Suppose now that $\nux$ is not a minimal $K$-fuzzy model of
  $(\prog, \tau)$. Hence, there exists some model $\nu'$ with
  $\nu' < \nux$. Let $\ybf$ be the solution of $\optk$ induced by $\nu'$. Since at least one atom is less true, and none are more true in $\nu'$ we have $\sum y_i < \sum x_i$, contradicting the optimality of $\xbf$ for $\optk$.
\end{proof}

Since the oblivious chase is polynomial with respect to data complexity (without existential quantification) and fi§nding optimal solutions of linear programs is famously polynomial~\cite{khachiyan1979polynomial}, our characterisation of minimal $K$-fuzzy models also directly reveals the complexity of \ktruth. Recall that for $K=1$, \ptime-hardness is inherited from classical Datalog (see~\cite{DBLP:journals/csur/DantsinEGV01}). 

\begin{corollary}
  Fix a rational $K\in (0,1]$. 
  \ktruth is in \ptime with respect to data complexity. Moreover, $1$-\textsc{Truth} is \ptime-complete in data complexity. 
\end{corollary}

To conclude this section we briefly note why it is important to use the oblivious chase as a basis for the ground rules that make up $\optk$. Consider the program $P(x), Q(x) \rightarrow S(x)$ and $R(x) \rightarrow S(x)$. In plain Datalog, $S(a)$ can be inferred simply if either $P(a)$ and $Q(a)$, or $R(a)$ are in the database. However, in the many-valued setting the truth of $S(a)$ depends on the truth of all three facts, $P(a)$, $Q(a)$, and $R(a)$ since it needs to be at least as true as $P(x) \conj Q(a)$ and at least as true $R(a)$. Technically this means $\optk$ needs to consider all possible paths of inferring each ground atom, and this exactly corresponds to the oblivious chase applying rules once per homomorphism from the body into the instance.

\subsection{Classical Behaviour for Certain Knowledge}
We are particularly interested in settings where one wants to  combine large amounts of certain knowledge, e.g., classical knowledge graphs and databases, with additional fuzzy information.
Going back to our example from the introduction, we have such a scenario if we want to reason on properties of the animals in the pictures by using pre-existing biological and zoological ontologies.
Ideally, we want inference on the ontologies to behave exactly as in the Datalog and only differ in behaviour when such \emph{certain knowledge} is combined with fuzzy observations and consequences of fuzzy observations. Note that particularly in the context of ontologies the question of fuzziness has also been studied extensively in the field of fuzzy description logics, see e.g., \cite{DBLP:conf/sum/BorgwardtP17,DBLP:journals/jair/StoilosSPTH07,DBLP:journals/ws/LukasiewiczS08}. Notably, the $\mathcal{EL}$ family of descrioption logics, which is closely related to rule based languages, has also been studied specifically with fuzzy semantics based on \luka logic and shown undecidable by~\cite{DBLP:conf/dlog/BorgwardtCP17}.

We show now that \mvdl indeed exhibits such ideal behaviour. In particular, when we consider Lemma~\ref{lem:crispub} but with a database consisting only of certain knowledge in $\tau$, i.e, ground atoms for which $\tau$ evaluates to $1$,  we can state the following.

\begin{theorem}
  \label{thm:certain}
  Let $\prog$, $\tau$ be a \mvdl program and database. Let $D^1$ be the classical database containing all $G\in \gatom$ where $\tau(G) =1$.
  If $\prog, \tau$ is $1$-satisfiable, 
  then for every ground atom $G$, $\progcrisp, D^1 \models G$ if and only if $\mu(G)=1$ where $\mu$ is the minimal $1$-fuzzy model of $(\prog,\tau)$.
\end{theorem}

Beyond the conceptual importance of Theorem~\ref{thm:certain}, the
statement also has significant practical consequences. Our
characterisation of minimal $K$-fuzzy models provides a clear algorithm for computing minimal $K$-fuzzy models via the oblivious chase and then solving $\optk$. Using the oblivious chase is important for our characterisation, since different ground rules with the same head may differ in how true their body is. In classical Datalog it is of course sufficient to infer the head once, which can require significantly less effort in practice. 
According to Theorem~\ref{thm:certain} (and by uniqueness of minimal fuzzy models) it is not necessary to compute the truth degree of those facts that are entailed by $\progcrisp, D^1$ via $\optk$.
An implementation can use this observation to produce less ground rules (i.e., not the full oblivious chase) and smaller equivalent versions of $\optk$ with less computational effort.
We conclude that \mvdl is well suited for combining large amounts of certain knowledge with fuzzy observations. Not only are the semantics on certain data equivalent to those of classical Datalog, but the computation of minimal fuzzy models also exhibits desirable properties in the presence of certain data.

\section{\mvpm}
\label{sec:mvpm}

Existential quantification in the head of rules provide a natural way of dealing with the common problem of incomplete or missing data by declaratively stating that certain facts must exist. Such usage is especially powerful in applications where rule-based reasoning is used for complex data analysis (see e.g., \cite{DBLP:conf/ijcai/BellomariniGPS17}). In the following, we introduce \mvpm as an extension of \mvdl with a focus on providing a fuzzy reasoning language that is useful for such applications.

To obtain semantics matching our intuition for such a system, we propose an alternative to the commonly studied semantics for existential quantification for \luka logic. We then identify certain \emph{preferred models} that exhibit desired conceptual and computational behaviour as the basis for reasoning in \mvpm.

\subsection{Strong Existential Quantification in \luka Semantics}
The semantics of existential quantification in \luka logic is commonly defined as $\nu(\exists \xbf \varphi(\xbf)) = \sup\{\nu(\varphi[\xbf/\cbf]) \mid \cbf \in \dom^{|x|}\}$
where $\varphi[\xbf/\cbf]$ is the substitution of variables $\xbf$ in $\varphi$ by constants $\cbf$ (cf.~\cite{DBLP:books/kl/Hajek98}).

However, with our focus set on practical applications, we propose alternative semantics for existential quantification that match the semantics of $\disj$. To contrast between the aforementioned semantics of $\exists$ we refer to our semantics as \emph{strong existential quantification}.
\[
  \nu(\exists \xbf\ \varphi(\xbf)) = \min\{1, \sum_{\cbf \in \dom^{|x|}} \nu(\varphi[\xbf/\cbf])\}
\]
We refer to $\Llog$ extended with strong existential quantification as $\Llog_\exists$ (following standard syntax of first order existential quantification).
Such semantics for quantification in \luka logic have also recently been studied more generally in a purely logical context by~\cite{10.1215/00294527-2021-0012}. 

\begin{definition}
  A \mvpm program is a set of $\Llog_\exists$ formulas that are either \mvdl rules or of the form
  \[
    B_1 \conj B_2 \conj \cdots \conj B_n \rightarrow \exists \xbf\ \varphi(\xbf, \ybf)
  \]
  where $B_i$ are relational atoms, $\ybf$ are the variables in the body, and $\varphi$ is formula that contains only a relational atom that uses all variables of $\xbf$.
\end{definition}

Note that we generally do not restrict the domain in the semantics of $\Llog$ formulas. Thus, as in \datalogpm, interpretations are not constrained to the active domain of $\prog$ and $\tau$ but allow for the introduction of new constants. As noted above, the introduction of new constants provides a natural mechanism to deductively handle missing and incomplete data. Our proposed use of strong existential quantification is motivated by its behaviour in cases of incomplete information as illustrated in the following example.

\begin{example}
  \label{ex:pm}
  We consider the following example rule, expressing that every company has a key person, to illustrate the difference between the two semantics for $\exists$.
  \[
    \begin{array}{ll}
      \mathit{Company}(y)& \rightarrow \exists x\ \mathit{KeyPerson}(x, y)
    \end{array}
  \]
  Suppose we have the following database, we know for certain that $\mathrm{Acme}$ is a company and  we are $0.8$ degrees confident that $\mathrm{Amy}$ is the key person of $\mathrm{Acme}$.
  \[
    \tau(\mathit{KeyPerson}(\mathrm{Amy}, \mathrm{Acme}))  = 0.8, \qquad \tau(\mathit{Company}(\mathrm{Acme})) = 1
  \]
  We discuss the case in a $1$-fuzzy model in the following. With existential semantics via the supremum of matching ground atoms, the first rule implies the existence of some new $\mathit{KeyPerson}(\mathrm{N}_1, \mathrm{Acme})$ with truth degree $1$. With strong existential quantification, the first rule implies only truth degree $0.2$ for some $\mathit{KeyPerson}(\mathrm{N}_1, \mathrm{Acme})$ since the known observation $\mathit{KeyPerson}(\mathrm{Amy}, \mathrm{Acme})$ already contributes $0.8$ degrees of truth to the head.

  Intuitively, the new constant $\mathrm{N}_1$ assumes the role of the unknown other object that possibly is the key person. In traditional semantics for $\exists$ we have to infer (in a $1$-fuzzy model) that $\mathrm{N}_1$ is certain to control $\mathrm{Acme}$.  In contrast, under strong existential semantics, the confidence in $\mathrm{N}_1$ being the key person is determined by the known observation that $\mathrm{Amy}$ might be the key person. 
\end{example}

\subsection{Preferred Models for \mvpm}
When considered in full generality, our proposed strong existential quantification semantics reveal a variety of theoretical challenges. To allow for $K$-fuzzy models $\nu$ where
$\{G \in \gatom \mid \nu(G)>0\}$ is not finite would require  the introduction of further technical complexity in the definition. Furthermore, it can be necessary for a model to introduce multiple new constants per ground rule to satisfy an existential quantification in certain situations as illustrated by the following example.
\begin{example}
  \label{ex:nulls}
  Consider database $\tau$ with
  \(
    \tau(S(a)) = 0.8, 
    \tau(T(a)) = 0.2
  \)
  and program $\prog$
  \[
    \begin{array}{ll}
      S(x) & \rightarrow \exists y\ P(x, y) \\
      P(x, y) & \rightarrow T(x)
    \end{array}
  \]
  Then there is no $1$-fuzzy model with only one new constant. However, there is a solution
  where the facts $P(a, n_1), P(a, n_2), P(a, n_3), P(a, n_4)$, with new constants $n_1,\dots,n_4$, are all assigned truth degree $0.2$.
\end{example}

Infinite models and the phenomenon illustrated by Example~\ref{ex:nulls} present
interesting topics for further theoretical study. For practical
applications we propose to avoid these corner cases and focus on
natural \emph{preferred models} that avoid these issues and provide
predictable outcomes.

Recall that the oblivious chase sequence introduces new constants exactly once for every possible true grounding of the body. This limitation on the number of new constants provides a natural balance between not overly restricting facts with null values in models while also avoiding extreme situations as in Example~\ref{ex:nulls} where arbitrarily many constants are introduced to satisfy the rule for a single ground body. 
We also restrict minimality of fuzzy models to only those ground atoms without labelled nulls. Minimising also the truth of atoms with labelled nulls would intuitively correspond to a kind of soft closed-world assumption, where we want to infer truth of statements under the assumption that the unknown information is minimal. 

\begin{definition}
  Let $\prog, \tau$ be an \mvpminst.
  Let $\mu$, $\nu$ be two $K$-fuzzy models of $(\prog, \tau)$. We say that $\mu \leq_{a} \nu$ if
  for all $G \in \activeatoms$ it holds that $\mu(G) \leq \nu(G)$.
\end{definition}

\begin{definition}
  We say that a $K$-fuzzy model $\nu$ of \mvpm instance $(\prog, \tau)$ has an \emph{oblivious base} if $\mathit{OLim}(\progcrisp, D_\tau)$ is finite, if for every $G\in\gatom$ it holds that $\nu(G) > 0$ only if $G \in \mathit{OLim}(\progcrisp, D_\tau)$.
  We say that $\nu$ is a \emph{preferred} $K$-fuzzy model if it is minimal with respect to the partial ordering $\leq_{a}$ of all $K$-fuzzy models with an oblivious base .
\end{definition}

This definition of preferred models captures the previously outlined intuition of which kinds of models we want to consider for practical applications. That is, we propose preferred models as the desired output of an \mvpm system. An \mvpminst does not necessarily have a unique preferred model. The question of how to compare preferred models among each other is left as an open question.

We now extend the definition of \optk from Section~\ref{sec:min} to program $\eoptk$ for an \mvpminst $\prog, \tau$. Let $\mathcal{G}$ and $\Gamma$ be as in the previous definition of $\optk$. For ground rules that are the result of grounding a rule without existential quantification, $\Vf$ is defined the same as it was for $\optk$.
Otherwise, let $\gamma$ be a ground rule obtained from grounding of an existential rule. Let $G_h$ be the atom in the head of $\gamma$ and let $N$ be the labelled nulls at the positions in $G_h$ that were existentially quantified in the original rule before grounding. We say that a ground atom $G$ \emph{matches} $G_h$ if $G$ can be obtained from $G_h$ by replacing the labelled nulls in $N$ by some other values of the domain.
Let $G_{\beta_1},\dots, G_{\beta_k}$ be the ground atoms in $\mathcal{G}$ that match $G_h$. We can then extend our previous definition of $\Vf$ to groundings of existential rules as
\[
  \Vf(\gamma) := \sum_{j=1}^\ell \left(1-x_{\alpha_j} \right)+  \sum_{i=1}^k x_{\beta_i}
\]
with variables $x_{\alpha_j}$ and $x_{\beta_i}$ corresponding, respectively, to the truth values of the ground atoms in the body and the head.
Then $\eoptk$ is the linear program defined the same as $\optk$, but with objective function
\begin{equation}
  \begin{array}{llr}
    \text{minimise} & \sum_{i=1}^n w_i x_i  \\
  \end{array}
  \label{eq:eoptk}
\end{equation}
where $x_i$ provides the truth value of $G_i$, $w_i= 1$ if $G_i \in \activeatoms$ and $w_i=0$ otherwise.
That is, ground atoms that contain nulls do not contribute to the objective. Recall from Example~\ref{ex:nulls} that this does not imply that they can be simply set to be fully true.

\begin{theorem}
  \label{thm:lamin}
  Let $\prog$, $\tau$ be a \mvdl program and database where $\mathit{OLim}(\progcrisp, D_\tau)$ is finite. The following two statements hold.
  \begin{enumerate}
  \item   If $(\prog, \tau)$ has a preferred $K$-fuzzy model, then $\eoptk$ is feasible.
  \item    For any  optimal solution $\xbf$ of $\eoptk$, we have that $\nu_\xbf$ is a preferred $K$-fuzzy model for $(\prog, \tau)$.
  \end{enumerate}
\end{theorem}

From our definition of preferred models we are naturally interested in identifying cases where the oblivious chase is finite. There exist syntactic restrictions for which we can show that this is always the case. In particular, we can identify programs where the oblivious chase is finite based on the notion of weakly acyclic programs as studied by~\cite{DBLP:journals/tcs/FaginKMP05}. Technical details of this question are discussed in \ref{sec:finchase}.

\section{Conclusion \& Outlook}
\label{sec:conclusion}
We have introduced \mvdl and \mvpm, two many-valued rule-based reasoning languages following \luka semantics. We provide theoretical foundations for \mvdl by characterising its minimal models in terms of an optimisation problem and showing that fuzzy fact entailment in \mvdl has the same complexity as fact entailment in Datalog.
For \mvpm we introduce alternative semantics for existential quantification in \luka logic that better match the popular usage of $\exists$ for handling missing data. As a first starting point for the practical use of \mvpm we introduce preferred models and show how to compute them.

We argue that \mvdl and its extensions can be effective languages for reasoning in settings with uncertain data. Therefore the development of a system based on the results given here is the natural next step to validate efficiency and effectiveness in comparison to other formalisms. This also entails the study of how the computation of minimal or preferred models can be further optimised, ideally requiring fewer ground rules than the methods described here.

Beyond practical considerations, a variety of natural theoretical questions have been left open. Prime among them are the questions of how fuzzy models of \mvdl behave with stratified negation, and how to handle the behaviour of \mvpm instances with an infinite oblivious chase. The latter is particularly intriguing since truth degrees effectively decreases monotonically along paths of derivation. We therefore suspect there to be relevant classes of problems where the oblivious chase is infinite but instances always permit finite fuzzy models.

Finally, there is an intriguing possibility of combining \mvdl semantics with the probabilistic semantics of Problog to reason about uncertainty and likelihood as two orthogonal concepts at the same time. 
\vspace{-1em}

\section*{Acknowledgements}
Stefano Sferrazza was supported by the Austrian Science Fund
(FWF):P30930. Georg Gottlob is a Royal Society Research Professor and acknowledges support by the Royal Society in this role through the  “RAISON DATA” project  (Reference No. RP\textbackslash{}R1\textbackslash{}201074). Matthias Lanzinger acknowledges support by the Royal Society  “RAISON DATA” project  (Reference No. RP\textbackslash{}R1\textbackslash{}201074).

\bibliography{0-main.bib}

\begin{thebibliography}{}

\bibitem[Achs and Kiss, 1995]{DBLP:journals/actaC/AchsK95}
{\sc Achs, {\'{A}}.} {\sc and} {\sc Kiss, A.} 1995.
\newblock Fuzzy extension of datalog.
\newblock {\em Acta Cybern.}, {\it 12}, 2, pp. 153--166.

\bibitem[Bach et~al., 2017]{DBLP:journals/jmlr/BachBHG17}
{\sc Bach, S.~H.}, {\sc Broecheler, M.}, {\sc Huang, B.}, {\sc and} {\sc
  Getoor, L.} 2017.
\newblock {Hinge-Loss Markov Random Fields and Probabilistic Soft Logic}.
\newblock {\em J. Mach. Learn. Res.}, {\it 18}, pp. 109:1--109:67.

\bibitem[Bellomarini et~al., 2017]{DBLP:conf/ijcai/BellomariniGPS17}
{\sc Bellomarini, L.}, {\sc Gottlob, G.}, {\sc Pieris, A.}, {\sc and} {\sc
  Sallinger, E.}
\newblock Swift logic for big data and knowledge graphs.
\newblock In {\em Proc. {IJCAI}} 2017, pp. 2--10. ijcai.org.

\bibitem[Bellomarini et~al., 2020]{DBLP:conf/ruleml/BellomariniLSS20}
{\sc Bellomarini, L.}, {\sc Laurenza, E.}, {\sc Sallinger, E.}, {\sc and} {\sc
  Sherkhonov, E.}
\newblock {Reasoning Under Uncertainty in Knowledge Graphs}.
\newblock In {\em RuleML+RR} 2020, pp. 131--139. Springer.

\bibitem[Bobillo and Straccia, 2016]{DBLP:journals/kbs/BobilloS16}
{\sc Bobillo, F.} {\sc and} {\sc Straccia, U.} 2016.
\newblock The fuzzy ontology reasoner fuzzydl.
\newblock {\em Knowl. Based Syst.}, {\it 95}, pp. 12--34.

\bibitem[Borgwardt et~al., 2017]{DBLP:conf/dlog/BorgwardtCP17}
{\sc Borgwardt, S.}, {\sc Cerami, M.}, {\sc and} {\sc Pe{\~{n}}aloza, R.}
\newblock {\L}ukasiewicz fuzzy {EL} is undecidable.
\newblock In {\em Proc. {DL}} 2017, volume 1879 of {\em {CEUR} Workshop
  Proceedings}. CEUR-WS.org.

\bibitem[Borgwardt and Pe{\~{n}}aloza, 2017]{DBLP:conf/sum/BorgwardtP17}
{\sc Borgwardt, S.} {\sc and} {\sc Pe{\~{n}}aloza, R.}
\newblock Fuzzy description logics - {A} survey.
\newblock In {\em Proc. {SUM}} 2017, volume 10564 of {\em Lecture Notes in
  Computer Science}, pp. 31--45. Springer.

\bibitem[Cal{\`{i}} et~al., 2013]{DBLP:journals/jair/CaliGK13}
{\sc Cal{\`{i}}, A.}, {\sc Gottlob, G.}, {\sc and} {\sc Kifer, M.} 2013.
\newblock {Taming the infinite chase: Query answering under expressive
  relational constraints}.
\newblock {\em JAIR}, {\it 48}, pp. 115--174.

\bibitem[Cal{\`{i}} et~al., 2012]{DBLP:journals/ws/CaliGL12}
{\sc Cal{\`{i}}, A.}, {\sc Gottlob, G.}, {\sc and} {\sc Lukasiewicz, T.} 2012.
\newblock {A general Datalog-based framework for tractable query answering over
  ontologies}.
\newblock {\em Journal of Web Semantics}, {\it 14}, pp. 57--83.

\bibitem[Cavallo and Pittarelli, 1987]{DBLP:conf/vldb/CavalloP87}
{\sc Cavallo, R.} {\sc and} {\sc Pittarelli, M.}
\newblock The theory of probabilistic databases.
\newblock In {\em VLDB'87} 1987, pp. 71--81. Morgan Kaufmann.

\bibitem[Dantsin et~al., 2001]{DBLP:journals/csur/DantsinEGV01}
{\sc Dantsin, E.}, {\sc Eiter, T.}, {\sc Gottlob, G.}, {\sc and} {\sc Voronkov,
  A.} 2001.
\newblock Complexity and expressive power of logic programming.
\newblock {\em {ACM} Comput. Surv.}, {\it 33}, 3, pp. 374--425.

\bibitem[Ebrahim, 2001]{DBLP:journals/fss/Ebrahim01}
{\sc Ebrahim, R.} 2001.
\newblock Fuzzy logic programming.
\newblock {\em Fuzzy Sets Syst.}, {\it 117}, 2, pp. 215--230.

\bibitem[Fagin et~al., 2005]{DBLP:journals/tcs/FaginKMP05}
{\sc Fagin, R.}, {\sc Kolaitis, P.~G.}, {\sc Miller, R.~J.}, {\sc and} {\sc
  Popa, L.} 2005.
\newblock {Data exchange: semantics and query answering}.
\newblock {\em Theor. Comput. Sci.}, {\it 336}, 1, pp. 89--124.

\bibitem[Fjellstad and Olsen, 2021]{10.1215/00294527-2021-0012}
{\sc Fjellstad, A.} {\sc and} {\sc Olsen, J.-F.} 2021.
\newblock {${\mathrm{IKT}^{\omega }}$ and Łukasiewicz-Models}.
\newblock {\em Notre Dame Journal of Formal Logic}, {\it 62}, 2, pp. 247 --
  256.

\bibitem[H{\'{a}}jek, 1998]{DBLP:books/kl/Hajek98}
{\sc H{\'{a}}jek, P.} 1998.
\newblock {\em {Metamathematics of Fuzzy Logic}}, volume~4 of {\em Trends in
  Logic}.
\newblock Kluwer.

\bibitem[Iranzo and S{\'{a}}enz{-}P{\'{e}}rez,
  2018]{DBLP:journals/tfs/IranzoS18}
{\sc Iranzo, P.~J.} {\sc and} {\sc S{\'{a}}enz{-}P{\'{e}}rez, F.} 2018.
\newblock A fuzzy datalog deductive database system.
\newblock {\em {IEEE} Trans. Fuzzy Syst.}, {\it 26}, 5, pp. 2634--2648.

\bibitem[Khachiyan, 1979]{khachiyan1979polynomial}
{\sc Khachiyan, L.~G.}
\newblock A polynomial algorithm in linear programming.
\newblock In {\em Doklady Akademii Nauk} 1979, volume 244, pp. 1093--1096.
  Russian Academy of Sciences.

\bibitem[Lee et~al., 2017]{DBLP:journals/corr/abs-1709-03806}
{\sc Lee, H.~S.}, {\sc Jung, H.}, {\sc Agarwal, A.~A.}, {\sc and} {\sc Kim, J.}
  2017.
\newblock Can deep neural networks match the related objects?: {A} survey on
  imagenet-trained classification models.
\newblock {\em CoRR}, {\it abs/1709.03806}.

\bibitem[Lukasiewicz and Straccia, 2008]{DBLP:journals/ws/LukasiewiczS08}
{\sc Lukasiewicz, T.} {\sc and} {\sc Straccia, U.} 2008.
\newblock Managing uncertainty and vagueness in description logics for the
  semantic web.
\newblock {\em J. Web Semant.}, {\it 6}, 4, pp. 291--308.

\bibitem[Maier et~al., 2018]{DBLP:books/mc/18/MaierTKW18}
{\sc Maier, D.}, {\sc Tekle, K.~T.}, {\sc Kifer, M.}, {\sc and} {\sc Warren,
  D.~S.}
\newblock Datalog: concepts, history, and outlook.
\newblock In {\em Declarative Logic Programming: Theory, Systems, and
  Applications} 2018. {ACM} / Morgan {\&} Claypool.

\bibitem[Mishkin et~al., 2017]{DBLP:journals/cviu/MishkinSM17}
{\sc Mishkin, D.}, {\sc Sergievskiy, N.}, {\sc and} {\sc Matas, J.} 2017.
\newblock Systematic evaluation of convolution neural network advances on the
  imagenet.
\newblock {\em Comput. Vis. Image Underst.}, {\it 161}, pp. 11--19.

\bibitem[Preining, 2010]{DBLP:conf/lpar/Preining10}
{\sc Preining, N.}
\newblock G{\"{o}}del logics - {A} survey.
\newblock In {\em Proc. {LPAR}} 2010, volume 6397 of {\em Lecture Notes in
  Computer Science}, pp. 30--51. Springer.

\bibitem[Raedt et~al., 2007]{DBLP:conf/ijcai/RaedtKT07}
{\sc Raedt, L.~D.}, {\sc Kimmig, A.}, {\sc and} {\sc Toivonen, H.}
\newblock Problog: {A} probabilistic prolog and its application in link
  discovery.
\newblock In {\em Proc. {IJCAI}} 2007, pp. 2462--2467.

\bibitem[Richardson and Domingos, 2006]{DBLP:journals/ml/RichardsonD06}
{\sc Richardson, M.} {\sc and} {\sc Domingos, P.~M.} 2006.
\newblock {Markov logic networks}.
\newblock {\em Mach. Learn.}, {\it 62}, 1-2, pp. 107--136.

\bibitem[Stoilos et~al., 2007]{DBLP:journals/jair/StoilosSPTH07}
{\sc Stoilos, G.}, {\sc Stamou, G.~B.}, {\sc Pan, J.~Z.}, {\sc Tzouvaras, V.},
  {\sc and} {\sc Horrocks, I.} 2007.
\newblock Reasoning with very expressive fuzzy description logics.
\newblock {\em J. Artif. Intell. Res.}, {\it 30}, pp. 273--320.

\bibitem[Suciu et~al., 2011]{DBLP:series/synthesis/2011Suciu}
{\sc Suciu, D.}, {\sc Olteanu, D.}, {\sc R{\'{e}}, C.}, {\sc and} {\sc Koch,
  C.} 2011.
\newblock {\em Probabilistic Databases}.
\newblock Synthesis Lectures on Data Management. Morgan {\&} Claypool
  Publishers.

\bibitem[Vojt{\'{a}}s, 2001]{DBLP:journals/fss/Vojtas01}
{\sc Vojt{\'{a}}s, P.} 2001.
\newblock Fuzzy logic programming.
\newblock {\em Fuzzy Sets Syst.}, {\it 124}, 3, pp. 361--370.

\end{thebibliography}

\newpage
\appendix
\section{Proof Details}
\label{sec:proofs}

\begin{proof}[Proof of Lemma~\ref{lem:crispub}]
  Let $G$ be a ground atom with $\nu(G)>0$.
  Since $\nu$ is minimal there must be a ground rule with $G$ in the head such that for every body atom $G_i \in \body(\gamma)$  also $\nu(G_i) > 0$. Consider the tree with $\gamma$ as the root and a child for every body atom of $\gamma$. At each child take again a rule where the respective body atom is in the head and the body is true to some degree greater than $0$.

  Repeating the process, we must arrive at a finite tree where all leaves must correspond to ground rules where all body atoms are defined in $\tau$ with truth greater than $0$. Otherwise, we could simply set $\nu(G')=0$ for all ground atoms $G'$ whose branches do not lead to such leaves and $\nu$ would still be a $K$-fuzzy model, contradicting its minimality.

  Viewing the same tree in the context of $\progcrisp,D_\tau$ we then see, going from the leaves upward that the bodies of all these rules must be true in all models. Thus, also ultimately $\progcrisp,D_\tau \models G$.
\end{proof}

\begin{proof}[Proof of Lemma~\ref{lem:opt}]
  Let $\xbf$ be a feasible solution of $\optk$.  Observe that $\Vf(\gamma) = \nux(\gamma)$ for every ground rule $\gamma$: Since, $\Vf(\gamma) \geq K$ for all $\gamma \in \Gamma$,  we also see that $\nux$ $K$-satisfies every rule in $\oground(\progcrisp, D_\tau)$. By definition of an oblivious chase sequence and since $\nux(G)=0$ for any ground atom that is not mentioned in $\mathit{OLim}(\progcrisp, D_\tau)$, any other ground rule $\gamma \not \in \Gamma$ is trivially $K$-satisfied since both head and body have truth $0$.

  Let $\nu$ be a $K$-fuzzy model of $(\prog, \tau)$. Let $\xbf$ be a solution of $\optk$ with $x_i = \nu(G_i)$ for all $G \in \mathcal{G}$. Since $\nu$ $K$-satisfies all ground rules, we in particular have that $\nu(\gamma) = \Vf(\gamma) \geq K$ for all $\gamma \in \Gamma$ and therefore $\xbf$ is feasible.
%
   The third statement follows immediately from combination of the first two.
 \end{proof}

\begin{proof}[Proof of Lemma~\ref{lem:cycle}]
  By Lemma~\ref{lem:crispub}, and the fact that the oblivious chase sequence constructs minimal models for Datalog, we also have that $\mathcal{G}' \subseteq \mathit{OLim}(\progcrisp, D_\tau)$.
  It follows by construction that every derived $G\in \mathcal{G}$ in the head of some ground rule in $\oground(\progcrisp, D_\tau)$.
  What is left to show is that at least one such rule is also $\nu$-tight.

  Suppose the statement is false and let $\delta > 0$ be the minimal $\nu$-gap of rules in $\Gamma$ whose body atoms are all not in $\mathcal{G}'$.
  Let $\Gamma'$ contain only the $\nu$-tight rules in $\oground(\progcrisp, D_\tau)$ that contain an atom from $\mathcal{G}'$ in the head. Since we assume the statement false, all of the rules in $\Gamma'$ have a body atom that is also in $\mathcal{G}'$.
  Now consider the truth assignment $\nu'$ defined as
  \[
    \nu'(G) =
    \begin{cases}
      \max\{0, \nu(G)-\delta\} & \text{if } G \in \mathcal{G}' \\
      \nux(G) & \text{otherwise}
    \end{cases}
  \]
 By assumption we have that all $\nu$-tight rules where $\mathcal{G}'$ occurs in the head also have some atom from $\mathcal{G}'$ in the body. Meaning that under $\nu'$ the truth of both head and body are decreased by at least $\delta$ (or until both are $0$). Hence, all rules in $\Gamma'$ are still $K$-satisfied. We can also see that all other rules that contain an atom in $\mathcal{G}'$ remain $K$-satisfied under $\nu'$. When an atom of $\mathcal{G}'$ occurs in the body of a rule, the implication can only become more true. Where an atom $G$ from $\mathcal{G}'$ occurs only in the head of a rule $\gamma \not \in \Gamma'$, we have that the $\nux$-gap of the rule is at least $\delta$. That is, 
  \[
    \nux(G) \geq \nux(\body(\gamma))-1+K+\delta
  \]
  and therefore also
  \[
    \nu'(G) \geq \nux(G)-\delta \geq \nux(\body(\gamma))-1+K \geq \nu'(\body(\gamma))-1+K
  \]
  We therefore see that $\nu'$ is a $K$-fuzzy model with $\nu' < \nu$ and we arrive at a contradiction.
\end{proof}

\begin{proof}[Proof of Theorem~\ref{thm:certain}]
  The implication from right to is a a special case of Lemma~\ref{lem:crispub}.
  For the other direction, we argue by induction over the oblivious chase sequence $D_0,D_1,\dots$ for $\progcrisp, D^1$, that if some ground atom $G$ is in $D_i$, then $\mu(G)=1$ in every $1$-fuzzy model $\mu$ of $(\prog, \tau)$. The base case $D_0 = D^1$ follows by definition of $D^1$. For the induction, suppose the claim holds up to step $i$. In the step from $D_i$ to $D_{i+1}$ either the ground atoms in both are the same (the oblivious application produced a ground head that was already in $D_i$) or there is a single new ground atom $G$ in $D_{i+1}$ but not in $D_i$.  Let $\gamma$ be the ground rule induced by the oblivious application in the step from $i$ to $i+1$. Since the rule was applicable, for all $G' \in \body(\gamma)$ we have  $G' \in D_i$ and by the induction hypothesis also $\mu(G')=1$. Hence, $\mu(\body(\gamma))=1$ and since $\gamma$ must be $1$-satisfied, we have that in every $1$-fuzzy model also $\mu(G)=1$.

  Recall, that $\mathit{OLim}(\progcrisp, D^1)$ is always a minimal model of $\progcrisp, D^1$ and thus $\progcrisp, D^1 \models G$ if and only if $G \in \mathit{OLim}(\progcrisp, D^1)$. And by the above induction argument $G \in \mathit{OLim}(\progcrisp, D^1)$ implies that $\mu(G)=1$.
\end{proof}

\begin{proof}[Proof of Theorem~\ref{thm:lamin}]
  First, observe that since $\mathit{OLim}(\progcrisp, D_\tau)$ is unique up to isomorphism, we can assume without loss of generality that there is no preferred model where some ground atom $G$ is true but not in the set $\mathcal{G}$ of ground atoms considered in the construction of $\eoptk$.

  Let $\nu$ be a preferred $K$-fuzzy model of $\prog,\tau$. Since $\nu$ has an oblivious base, we have that $\nu(G) > 0$ only if $G \in \mathcal{G}$. Let $\xbf$ be the solution of $\eoptk$ where $x_i = \nu(G_i)$ for all $G_i\in \mathcal{G}$. Then, by definition of $\Vf$ we immediately see $\Vf(\gamma) = \nu(\gamma)$ for every ground rule $\gamma$. Hence, if all ground rules are $K$-satisfied by $\nu$, then all constraints $\Vf(\gamma) \geq K$ are satisfied and $\nu$ is feasible.
  
  For the second statement, assume that $\xbf$ is an optimal solution of $\eoptk$ but $\nux$ is not preferred. The only way $\nux$ can not be preferred is if it is not active minimal, i.e., there is some $K$-fuzzy model $\mu$ where for all $G\in\activeatoms$, $\mu(G) \leq \nux(G)$ and for at least one $G' \in \activeatoms$, $\mu(G') < \nux(G')$. Since both $\mu$ and $\nux$ have an oblivious base (which is unique up to isomorphism) we see that every $G \in \gatom$ where $\mu(G)>0$ is also in $\mathcal{G}$. Hence, it is straightforward to construct a feasible solution $\xbf'$ from $\mu$ for which the objective function is strictly lower than for $\mathbf{x}$, a contradiction.
\end{proof}

\section{Finiteness in the Oblivious Chase}
\label{sec:finchase}
Since we are interested particularly in instances where the chase is finite, it is of interest to identify fragments where this is always the case. The most prominent condition for which the finiteness of the chase in \datalogpm is studied, is the restriction to \emph{weakly acyclic programs} as first introduced by~\cite{DBLP:journals/tcs/FaginKMP05}.
Let $\prog$ be a set of rules with over the relational language $\sigma$. We first define the \emph{dependency graph} for $\prog$ as the graph with vertices
\(
  \{ (R, i) \mid R \in \sigma, i \in \{1,\dots,\#(R)\} \}
\)
 We say a variable $x$ is in position $(R, i)$ if $x$ is at the $i$-th place of a relation with name $R$.
There is a (normal) edge from vertex $(R, i)$ to $(S, j)$ exactly if there is a rule in $\prog$ where the same variable $x$ occurs at position $(R,i)$ in the body and at position $(S,j)$ in the head of the rule. There is a \emph{special edge} from $(R,i)$ to $(S,j)$ exactly if there is an existential rule in $\prog$, such that there is a variable $x$ that $x$ occurs in both the body and head, $x$ is in position $(R,i)$ of the body, and an existentially quantified $y$ is in position $(S,j)$ in the head.
We say that a program is \emph{weakly acyclic} if its dependency graph has no cycle that passes a special edge. 

However, the standard definition of weakly acyclic is particular to the restricted chase, which does not fit our setting. This becomes apparent when considering the simple (weakly acyclic) program
\[
  P(x) \rightarrow \exists y\ P(y).
\]
The oblivious chase is infinite in this case, since every new instantiation of $y$ generates a new homomorphism that satisfies the body. Since $x$ does not occur in the head at all, the dependency graph used in the standard definition of weakly-acyclic does not have any edges and is thus weakly acyclic.

Let the \emph{variable expansion} $\ve(\prog)$ of $\prog$ be the program obtained by the following rewriting.
Without loss of generality every existential rule is of the form $\varphi(\xbf) \rightarrow \exists \ybf\  R(\ybf,\xbf')$, where $\varphi$ is the body formula, $R$ is a relation symbol and $\xbf'$ are those (free) variables that occur in both the body and the head.  Let $\xbf^*$ be the free variables that only occur in $\varphi$ but not in the head. Then, for every such rule we replace it by the two rules 
\[
    \begin{array}{ll}
         \varphi(\xbf) &\rightarrow \exists \ybf\ R^*(\ybf, \xbf', \xbf^*) \\
         R^*(\ybf, \xbf', \xbf^*) &\rightarrow R(\ybf, \xbf')
    \end{array}
\]
to obtain $\ve(\prog)$. Intuitively, the variable expansion reveals those particular cycles that are harmless in the restricted chase but lead to infinite sequences in the oblivious chase in the standard definition of weak acyclicity.
For example, the variable expansion of our simple example above would thus be the (no longer weakly acyclic) program
\[
 \begin{array}{ll}
      P(x) & \rightarrow \exists y\ P^*(y, x)  \\
      P^*(y,x) & \rightarrow P(y) 
 \end{array}
\]
 
The following theorem then follows by similar argument as originally given for the restricted chase \cite[Theorem 3.9]{DBLP:journals/tcs/FaginKMP05} by additionally observing that nulls can never propagate to a position in the dependency graph with lower rank, where the rank of a position is the maximal number of special edges on an incoming path. One can then inductively bound the number of possible groundings of the body to consequently bound the number of nulls generated in the head for each position.

\begin{theorem}
  Let $\prog, \tau$ be a \mvpminst
  program.  If $\ve(\prog)$ is weakly acyclic, then $\mathit{OLim}(\progcrisp, D_\tau)$ is finite and of polynomial size in data complexity.
\end{theorem}

\end{document}